\documentclass[11pt,a4paper]{article}

\usepackage{fullpage}
\usepackage[utf8x]{inputenc}
\usepackage{amsmath, amsthm}
\usepackage{wrapfig,graphicx,amssymb,textcomp,array,amsmath}
\usepackage{algpseudocode} 
\usepackage{enumerate}
\usepackage{enumitem}
\usepackage{multirow}
\usepackage{tabularx}
\algtext*{EndWhile}
\algtext*{EndIf}
\usepackage{algorithm}
\usepackage{color}
\usepackage{authblk}

\newcommand{\CH}[1]{\text{$CH(#1)$}}

\newcommand{\K}[2]{\text{$K(#1,#2)$}}

\newcommand{\require}{\textbf{Input: }}
\newcommand{\ensure}{\textbf{Output: }}
\newcommand{\setup}{\textsf{plane-tree}}
\newcommand{\procA}{\textsf{proc1}}
\newcommand{\procB}{\textsf{proc2}}
\newcommand{\passrefine}{\textsf{pair-refine}}

\newcommand{\dgT}[2]{d_{#1}(#2)}
\newcommand{\polylog}{{\rm polylog}}

\title{Plane Bichromatic Trees of Low Degree
\thanks{Research supported by NSERC.}
}

\author{
Ahmad Biniaz\qquad 
Prosenjit Bose\qquad
Anil Maheshwari\qquad 
Michiel Smid}
\affil{School of Computer Science, Carleton University, Ottawa, Canada.}

\date{\today}
\newtheorem{lemma}{Lemma}

\newtheorem{conjecture}{Conjecture}
\newtheorem{theorem}{Theorem}
\newtheorem{observation}{Observation}
\newtheorem*{problem*}{Problem}
\begin{document}

\maketitle

\begin{abstract}
Let $R$ and $B$ be two disjoint sets of points in the plane such that $|B|\leqslant |R|$, and no three points of $R\cup B$ are collinear. We show that the geometric complete bipartite graph $\K{R}{B}$ contains a non-crossing spanning tree whose maximum degree is at most $\max\left\{3, \left\lceil \frac{|R|-1}{|B|}\right\rceil + 1\right\}$; this is the best possible upper bound on the maximum degree. This solves an open problem posed by Abellanas et al. at the Graph Drawing Symposium, 1996.
\end{abstract}
\section{Introduction}
\label{introduction-section}

Let $R$ and $B$ be two disjoint sets of points in the plane. We assume that the points in $R$ are colored red and the points in $B$ are colored blue. We assume that $R\cup B$ is in {\em general position}, i.e., no three points of $R\cup B$ are collinear. The {\em geometric complete bipartite graph} $\K{R}{B}$ is the graph whose vertex set is $R\cup B$ and whose edge set consists of all the straight-line segments connecting a point in $R$ to a point in $B$. A {\em bichromatic tree} on $R\cup B$ is a spanning tree in $\K{R}{B}$. A {\em plane bichromatic tree} is a bichromatic tree whose edges do not intersect each other in their interior. A {\em $d$-tree} is defined to be a tree whose maximum vertex degree is at most $d$. 

If $R\cup B$ is in general position, then it is possible to find a plane bichromatic tree on $R\cup B$ as follows. Take any red point and connect it to all the blue points. Extend the resulting edges from the blue endpoints to partition the plane into cones. Then, connect the remaining red points in each cone to a suitable blue point on the boundary of that cone without creating crossings. This simple solution produces trees possibly with large vertex degree. In this paper we are interested in computing a plane bichromatic tree on $R\cup B$ whose maximum vertex degree is as small as possible. This problem was first mentioned by Abellanas et al.~\cite{Abellanas1996} in the Graph Drawing Symposium in 1996:

\begin{problem*}
Given two disjoint sets $R$ and $B$ of points in the plane, with $|B|\leqslant |R|$, find a plane bichromatic tree on $R\cup B$ having maximum degree $O(|R|/|B|)$.
\end{problem*}

Assume $|B|\leqslant |R|$. Any bichromatic tree on $R\cup B$ has $|R|+|B|-1$ edges. Moreover, each edge is incident on exactly one blue point. Thus, the sum of the degrees of the blue points is $|R|+|B|-1$. This implies that any bichromatic tree on $R\cup B$ has a blue point of degree at least $\frac{|R|+|B|-1}{|B|}=\frac{|R|-1}{|B|}+1$. Since the degree is an integer, $\left\lceil \frac{|R|-1}{|B|}\right\rceil + 1$ is the best possible upper bound on the maximum degree.

\begin{wrapfigure}{r}{0.21\textwidth}
  \begin{center}
\vspace{-20pt}
\includegraphics[width=.18\textwidth]{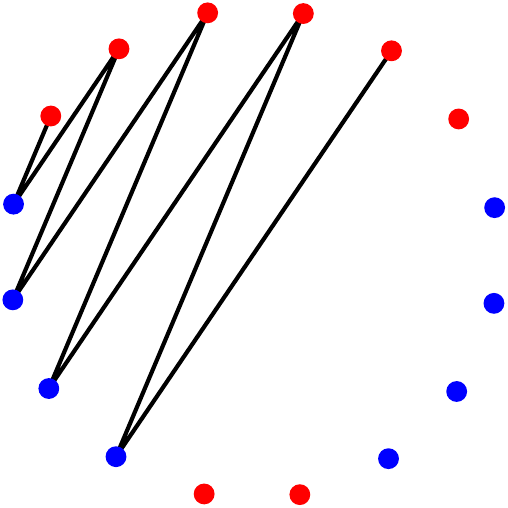}
  \end{center}
\vspace{-20pt}
\end{wrapfigure}
For cases when $|R|=|B|$ or $|R|=|B|+1$ it may not always be possible to compute a plane bichromatic tree of degree $\left\lceil \frac{|R|-1}{|B|}\right\rceil + 1=2$, i.e., a plane bichromatic path; see the example in the figure on the right which is borrowed from~\cite{Abellanas1999}; by adding one red point to the top red chain, an example for the case when $|R|=|B|+1$ is obtained. 
In 1998, Kaneko~\cite{Kaneko1998} posed the following conjecture.
\begin{conjecture}[Kaneko~\cite{Kaneko1998}]
\label{conj1}
Let $R$ and $B$ be two disjoint sets of points in the plane such that $|B|\leqslant|R|$ and $R\cup B$ is in general position, and let $k=\left\lceil\frac{|R|}{|B|}\right\rceil$. Then, there exists a plane bichromatic tree on $R\cup B$ whose maximum vertex degree is at most $k+2$.
\end{conjecture}

Kaneko~\cite{Kaneko1998} posed the following sharper conjecture for the case when $k\geqslant 2$.

\begin{conjecture}[Kaneko~\cite{Kaneko1998}]
\label{conj2}
Let $R$ and $B$ be two disjoint sets of points in the plane such that $|B|\leqslant|R|$ and $R\cup B$ is in general position, and let $k=\left\lceil\frac{|R|}{|B|}\right\rceil$. If $k\geqslant2$, then there exists a plane bichromatic tree on $R\cup B$ whose maximum vertex degree is at most $k+1$.
\end{conjecture}

\subsection{Previous Work}
\label{previous-work-section}
Assume $|B|\leqslant |R|$ and let $k=\left\lceil \frac{|R|}{|B|}\right\rceil$. Abellanas et al.~\cite{Abellanas1999} proved that there exists a plane bichromatic tree on $R\cup B$ whose maximum vertex degree is $O(k+\log |B|)$. Kaneko~\cite{Kaneko1998} showed how to compute a plane bichromatic tree of maximum degree $3k$. 

Kaneko~\cite{Kaneko1998} proved Conjecture~\ref{conj1} for the case when $|R|=|B|$, i.e., $k=1$; specifically he showed how to construct a plane bichromatic tree of maximum degree three. In~\cite{Kaneko2003} the authors mentioned that Kaneko proved Conjecture~\ref{conj1} for the case when $|R|\neq |B|$. However, we have not been able to find any written proof for this conjecture. Moreover, Kano~\cite{Kano2015communication} confirmed that he does not remember the proof.

Abellanas et al.~\cite{Abellanas1999} considered the problem of computing a low degree plane bichromatic tree on some restricted point sets. They proved that if $R\cup B$ is in convex position and $|R|=k|B|$, with $k\geqslant 1$, then $R\cup B$ admits a plane bichromatic tree of maximum degree $k+2$. If $R$ and $B$ are linearly separable and $|R| = k|B|$, with $k\geqslant 1$, they proved that $R\cup B$ admits a plane bichromatic tree of maximum degree $k+1$. They also obtained a degree of $k+1$ for the case when $B$ is the convex hull of $R\cup B$ (to be more precise, $B$ is equal to the set of points on the convex hull of $R\cup B$). To the best of our knowledge, neither Conjecture~\ref{conj1} nor Conjecture~\ref{conj2} has been proven for points in general position where $|R|\neq|B|$.

The existence of a spanning tree (not necessarily plane) of low degree in a bipartite graph is also of interest. If $|B|\leqslant|R|\leqslant k|B|+1$, with $k\geqslant 1$, then $\K{R}{B}$ has a spanning $(k+1)$-tree which can be computed as follows. Partition the points of $R$ into non-empty sets $R_1,\dots, R_{|B|}$ each of size at most $k$, except possibly one set which is of size $k+1$; let $R_{1}$ be that set. Let $B=\{b_1,\dots, b_{|B|}\}$. Let $S_1, \dots, S_{|B|}$ be the set of stars obtained by connecting the points in $R_i$ to $b_i$. To obtain a $(k+1)$-tree we connect $S_i$ to $S_{i+1}$ by adding an edge between a red point in $S_i$ to the only blue point in $S_{i+1}$ for $i=1,\dots,|B|-1$. Kano~et~al.~\cite{Kano2015} considered the problem of computing a spanning tree of low degree in a (not necessarily complete) connected bipartite graph $G$ with bipartition $(R,B)$. They showed that if $|B|\leqslant |R|\leqslant k|B|+1$, with $k\geqslant 1$, and $\sigma(G)\geqslant |R|$, then $G$ has a spanning $(k+1)$-tree, where $\sigma(G)$ denotes the minimum degree sum of $k+1$ independent vertices of $G$. 

The problem of computing a plane tree of low degree on multicolored point sets (with more than two colors) is also of interest, see~\cite{Kano2013a, Biniaz2015}. Let $P$ be a set of colored points in the plane in general position. Let $\{P_1, \dots, P_m\}$ be the partition of $P$ where the points of $P_i$ have the same color that is different from the color of the points of $P_j$ for $i\neq j$. Kano et al.~\cite{Kano2013a} showed that if $|P_i|\leqslant \left\lceil \frac{|P|}{2}\right\rceil$, for all $i=1,\dots, m$, then there exists a plane colored $3$-tree on $P$; a colored tree is a tree where the two endpoints of every edge have different colors. Biniaz et al.~\cite{Biniaz2015} presented algorithms for computing a plane colored $3$-tree when $P$ is in the interior of a simple polygon and every edge is the shortest path between its two endpoints in the polygon.

A related problem is the non-crossing embedding of a given tree into a given point set in the plane. 
It is known that every tree $T$ with $n$ nodes can be embedded into any set $P$ of $n$ points in the plane in general position~(see Lemma 14.7 in~\cite{Pach1995}). If $T$ is rooted at a node $r$, then the rooted-tree embedding problem asks if $T$ can be embedded in $P$ with $r$ at a specific point $p\in P$. This problem which was originally posed by Perles, is answered in the affirmative by Ikebe et al.~\cite{Ikebe1994}. A related result by A. Tamura and Y. Tamura~\cite{Tamura1992} is that, given a point set $P=\{p_1,\dots, p_n\}$ in the plane in general position and a sequence $(d_1,\dots,d_n)$ of positive integers with $\sum_{i=1}^{n} d_i=2n-2$, then $P$ admits an embedding of some tree $T$ such that the degree of $p_i$ is $d_i$. Optimal algorithms for solving the above problems can be found in~\cite{Bose1997}.

Another related problem is to compute a plane bichromatic Euclidean minimum spanning tree on $R\cup B$. This problem is NP-hard~\cite{Borgelt2009}. The best polynomial-time algorithm known so far for this problem, has approximation factor $O(\sqrt{n})$, where $n$ is the total number of points~\cite{Borgelt2009}.   

\subsection{Our Results}
\label{our-results-section}
In Section~\ref{k+2-section} we prove Conjecture~\ref{conj1} for the case when $|R|=k|B|$, with $k\geqslant 1$. The proof\textemdash which is very simple\textemdash is based on a result of Bespamyatnikh et al.~\cite{Bespamyatnikh2000} and a recent result of Hoffmann et al.~\cite{Hoffmann2014}. The core of our contribution is in Section~\ref{k+1-section}, where we partially prove Conjecture~\ref{conj2}: If $|R|=k|B|$, with $k\geqslant 2$, and $R\cup B$ is in general position, then there exists a plane bichromatic tree on $R\cup B$ whose maximum degree is $k+1$. We present a constructive proof for obtaining such a tree. We prove the full Conjecture~\ref{conj2} in Section~\ref{general-section}; the proof is again constructive and is based on the algorithm of Section~\ref{k+1-section}. Then, in Section~\ref{contribution-section} we combine the results of Sections~\ref{k+1-section} and~\ref{general-section} to show the following theorem that is even sharper than Conjecture~\ref{conj2}.

\begin{theorem}
\label{thr0}
Let $R$ and $B$ be two disjoint sets of points in the plane such that $|B|\leqslant|R|$ and $R\cup B$ is in general position, and let $\delta=\left\lceil \frac{|R|-1}{|B|}\right\rceil$. Then, there exists a plane bichromatic tree on $R\cup B$ whose maximum vertex degree is at most $\max\{3, \delta + 1\}$; this is the best possible upper bound on the maximum degree.
\end{theorem}

As we will see, the proof of Theorem~\ref{thr0} is simpler for $\delta\geqslant 4$. However, for smaller values of $\delta$, the proof is much more involved.

\section{Plane Bichromatic $(k+2)$-trees}
\label{k+2-section}
In this section we show that if $|R|=k |B|$, with $k\geqslant 1$, then there exists a plane bichromatic tree on $R\cup B$ whose maximum vertex degree is at most $k+2$. We make use of the following two theorems.
\begin{theorem}[Bespamyatnikh, Kirkpatrick, and Snoeyink~\cite{Bespamyatnikh2000}]
\label{eqsubdivision-thr}
Let $g\geqslant 1$ be an integer. Given $gn$ red and $gm$ blue points in the plane in general
position, there exists a subdivision of the plane into $g$ convex regions each of which contains $n$ red and $m$ blue points.
\end{theorem}

\begin{theorem}[Hoffmann and T{\'{o}}th~\cite{Hoffmann2014}]
\label{Hoffmann-thr}
Every disconnected plane bichromatic geometric graph with no isolated vertices can be augmented (by adding edges) into a connected plane bichromatic geometric graph such that the degree of each vertex increases by at most two.
\end{theorem}

\begin{wrapfigure}{r}{0.21\textwidth}
  \begin{center}
\vspace{-20pt}
\includegraphics[width=.18\textwidth]{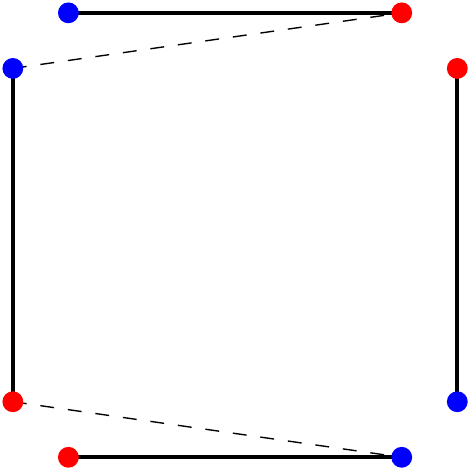}
  \end{center}
\vspace{-20pt}
\end{wrapfigure}
The subdivision in Theorem~\ref{eqsubdivision-thr} is known as an ``Equitable Subdivision'' of the plane. The result of Theorem~\ref{Hoffmann-thr} is optimal, as there are disconnected plane bichromatic geometric graphs with no isolated vertices that cannot be augmented into a connected plane bichromatic geometric graph while increasing each vertex degree by less than two; see the example in the figure on the right.  
\begin{theorem}
\label{th1}
Let $R$ and $B$ be two disjoint sets of points in the plane such that $|R|=k|B|$, with $k\geqslant 1$, and $R\cup B$ is in general position. Then, there exists a plane bichromatic tree on $R\cup B$ whose maximum vertex degree is at most $k+2$.
\end{theorem}
\begin{proof}
By applying Theorem~\ref{eqsubdivision-thr} with $g=|B|$ we divide the plane into $|B|$ convex regions each of which contains one blue point and $k$ red points.
In each region we obtain a star by connecting the $k$ red points in that region to its only blue point. Each star is a bichromatic tree of maximum degree $k$. By Theorem~\ref{Hoffmann-thr} we connect the stars to form a plane bichromatic tree of degree at most $k+2$.
\end{proof}

\section{Plane Bichromatic $(k+1)$-trees}
\label{k+1-section}
In this section we prove Conjecture~\ref{conj2} for the case when $|R|=k|B|$ and $k\geqslant 2$:

\begin{theorem}
\label{thr2}
Let $R$ and $B$ be two disjoint sets of points in the plane, such that $|R|=k|B|$, with $k\geqslant 2$, and $R\cup B$ is in general position. Then, there exists a plane bichromatic tree on
$R\cup B$ whose maximum vertex degree is at most $k + 1$.
\end{theorem}

Note that any bichromatic tree on $R\cup B$ has $|B|+|R|-1=(k+1)|B|-1$ edges. Since each edge is incident to exactly one blue point, the sum of the degrees of the blue points is $(k+1)|B|-1$. This implies the following observation:

\begin{observation}
 \label{obs1}
Let $R$ and $B$ be two disjoint sets of points in the plane such that $|R|=k|B|$, with $k\geqslant 1$ is an integer. Then, in any bichromatic $(k+1)$-tree on $R\cup B$, one point of $B$ has degree $k$ and each other point of $B$ has degree $k+1$.
\end{observation}

In order to prove~Theorem~\ref{thr2} we provide some notation and definitions. Let $P$ be a set of points in the plane. We denote by $\CH{P}$ the convex hull of $P$. For two points $p$ and $q$ in the plane, we denote by $(p,q)$ the line segment whose endpoints are $p$ and $q$. Moreover, we denote by $\ell(p,q)$ the line passing through $p$ and $q$. For a node $p$ in a tree $T$ we denote by $\dgT{T}{p}$ the degree of $p$ in $T$. Let $p$ be a vertex of $\CH{P}$. The {\em radial ordering} of $P-\{p\}$ around $p$ is obtained as follows. Let $p_1$ and $p_2$ be the two vertices of $\CH{P}$ adjacent to $p$ such that the clockwise angle $\angle p_1pp_2$ is less than $\pi$. For each point $q$ in $P-\{p\}$, define its angle around $p$\textemdash with respect to $p_1$\textemdash to be the clockwise angle $\angle p_1pq$. Then the desired radial ordering is obtained by ordering the points in $P-\{p\}$ by increasing angle around $p$.

We start by proving two lemmas that play an important role in the proof of Theorem~\ref{thr2}.

\begin{figure}[htb]
  \centering
  \includegraphics[width=.5\columnwidth]{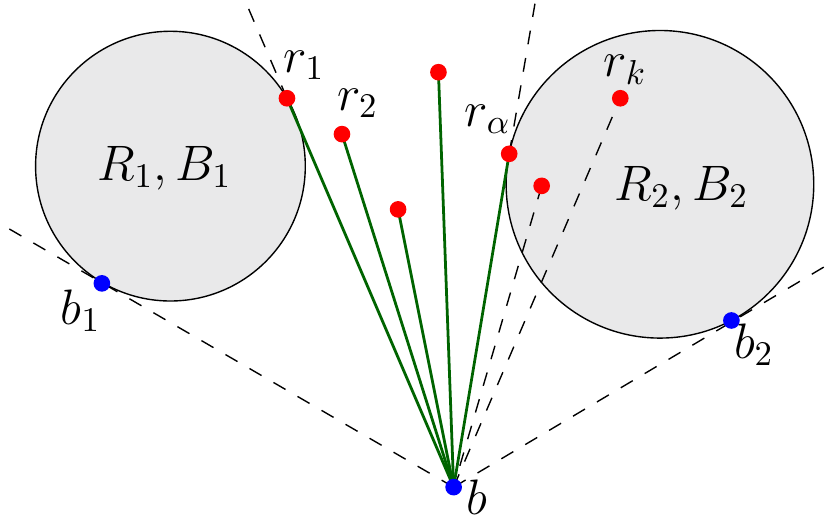}
 \caption{$k$ consecutive red points in the radial ordering of the points around $b$.}
  \label{convex-3blue-fig}
\end{figure}

\begin{lemma}
\label{convex-3blue}
Let $R$ and $B$ be two sets of red and blue points in the plane, respectively, such that $|B|\geqslant 1$, $k(|B|-1)< |R|\leqslant k|B|$, with $k\geqslant 2$, and $R\cup B$ is in general position. Let $b_1,b,b_2$ be blue points that are counter clockwise consecutive on $\CH{R\cup B}$. Then, in the radial ordering of $R\cup B-\{b\}$ around $b$, there are $\alpha=|R|-k(|B|-1)$ consecutive red points, $r_1,\dots,r_\alpha$, such that $|R_1|=k|B_1|+1$ and $|R_2|=k|B_2|+1$, where $R_1$ $($resp. $B_1$$)$ is the set of red points $($resp. blue points$)$ of $R\cup B-\{b\}$ lying on or to the left of $\ell(b,r_1)$, and $R_2$ $($resp. $B_2$$)$ is the set of red points $($resp. blue points$)$ of $R\cup B-\{b\}$ lying on or to the right of $\ell(b,r_\alpha)$.
\end{lemma}
\begin{proof}
By a suitable rotation of the plane, we may assume that $b$ is the lowest point of $\CH{R\cup B}$, and $b_1$ (resp. $b_2$) is to the left (resp. right) of the vertical line passing through $b$. Note that $b_1$ is the first point and $b_2$ is the last point in the clockwise radial ordering of $R\cup B-\{b\}$ around $b$. See Figure~\ref{convex-3blue-fig}. 
We define the function $f$ as follows: For every point $x$ in this radial ordering, 
\begin{align*}
f(x)= & k\cdot (\text{the number of points of $B-\{b\}$ lying on or to the left of $\ell(b,x)$})\\
& -(\text{the number of points of $R$ lying on or to the left of $\ell(b,x)$}).
\end{align*}
Based on this definition, we have $f(b_1)= k\geqslant 2$ and $f(b_2)=k(|B|-1)-|R|=-\alpha \leqslant -1$. Along this radial ordering, the value of $f$ changes by $+k$ at every blue point and by $-1$ at every red point. Since $f(b_1)>0>f(b_2)$, there exists a point in the radial ordering for which $f$ equals 0. Let $v$ be the last point in the radial ordering where $f(v)=0$. Since $b_2$ increases $f$ by $+k$ and $f(b_2)\leqslant -1$, there are at least $k+2$ points strictly after $v$ in the radial ordering. Let $S=(r_1, \dots, r_{k})$ be the sequence of $k$ points strictly after $v$ in the radial ordering. 

{\em Claim: The points of $S$ are red.}
Assume that $r_i$ is blue for some $1\leqslant i\leqslant k$. Then, $r_i$ changes $f$ by $+k$. Since $f$ decreases only at red points, the value of $f(r_i)$ is minimum when all points $r_1,\dots, r_{i-1}$ are red. Thus,
\begin{align*}
f(r_i)&\geqslant f(v) - (i-1) + k\\
&= k-i+1\\
&>0.
\end{align*}
Since $f(r_i)>0>f(b_2)$, there exists a point between $r_i$ and $b_2$ in the radial ordering for which $f$ equals 0. This contradicts the fact that $v$ is the last point in the radial ordering with $f(v)=0$. This proves the claim.

Thus, each $r_i\in S$ is red. Moreover, $f(r_i)=-i$. We show that the subsequence $S'=(r_1,\dots, r_\alpha)$ of $S$ satisfies the statement of the lemma; note that, by definition, $\alpha\leqslant k$.
Having $r_1$ and $r_\alpha$, we define $R_1$, $B_1$, $R_2$ and $B_2$ as in the statement of the lemma. See~Figure~\ref{convex-3blue-fig}. By definition of $f$, we have $f(r_1)=k|B_1|-|R_1|=-1$, and hence $|R_1|=k|B_1|+1$. Moreover, 
\begin{align*}
|R_2| &= |R|-|R_1|-|S'| +2\\
&= |R|-(k|B_1|+1)-(|R|-k(|B|-1)) +2\\
&= k(|B|-|B_1|-1)+1\\
 &= k|B_2|+1,
\end{align*}
where $|S'|$ is the number of elements in the sequence $S'$. Note that $R_2=(R-(R_1\cup S'))\cup\{r_\alpha\}$. Since $r_1$ belongs to both $R_1$ and $S'$, and $r_\alpha$ belongs to $R_2$, the term ``$+2$'' in the first equality is necessary (even for the case when $r_1=r_\alpha$). The last equality is valid because $B_2=B-(B_1\cup \{b\})$. This completes the proof of the lemma. 
\end{proof}

\begin{figure}[htb]
  \centering
\setlength{\tabcolsep}{0in}
  $\begin{tabular}{cc}
 \multicolumn{1}{m{.5\columnwidth}}{\centering\includegraphics[width=.35\columnwidth]{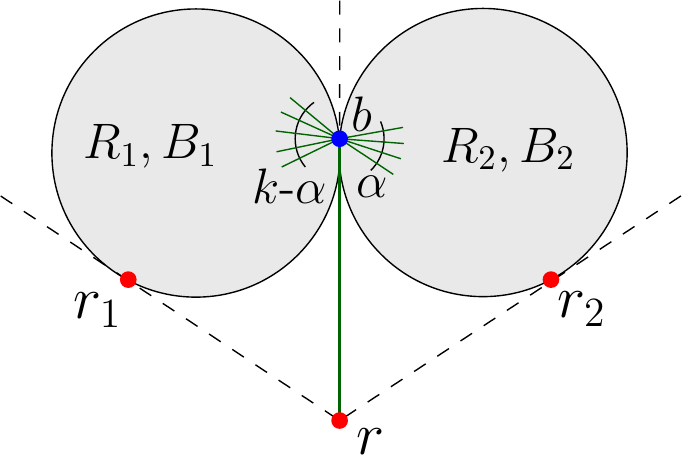}}
&\multicolumn{1}{m{.5\columnwidth}}{\centering\includegraphics[width=.37\columnwidth]{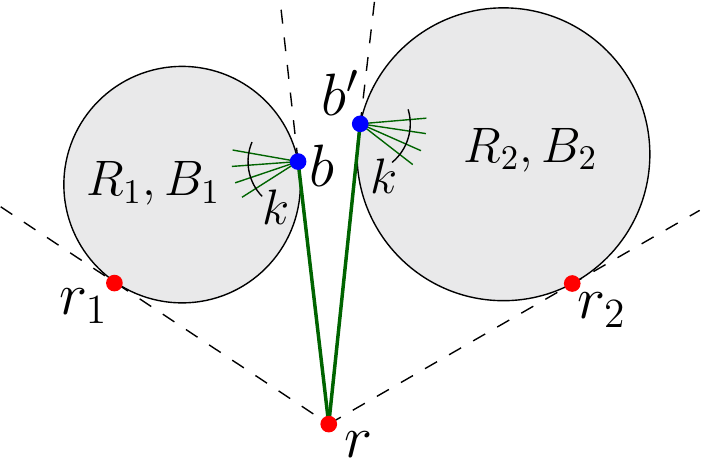}}\\
(a) & (b)
\end{tabular}$
  \caption{(a) $0<f(b)=\alpha<k$, (b) $f(b)=0$, and $b, b'$ are consecutive blue points.}
\label{convex-3red-fig}
\end{figure}

\begin{lemma}
\label{convex-3red}
Let $R$ and $B$ be two sets of red and blue points in the plane, respectively, such that $|B|\geqslant 1$, $|R|=k|B|+1$, with $k\geqslant 2$, and $R\cup B$ is in general position. Let $r_1,r,r_2$ be red points that are counter clockwise consecutive on $\CH{R\cup B}$. Then, one of the following statements holds: 
\begin{enumerate}
\item There exists a blue point $b$ in the radial ordering of $(R\cup B)-\{r\}$ around $r$, such that $|R_1|=k(|B_1|-1)+k-\alpha$ and $|R_2|=k(|B_2|-1)+\alpha$, where $R_1$ $($resp. $B_1$$)$ is the set of red points $($resp. blue points$)$ of $(R\cup B)-\{r\}$ lying on or to the left of $\ell(r,b)$, $R_2$ $($resp. $B_2$$)$ is the set of red points $($resp. blue points$)$ of $(R\cup B)-\{r\}$ lying on or to the right of $\ell(r,b)$, and $0<\alpha<k$.
\item There exist two consecutive blue points $b$ and $b'$ in the radial ordering of $(R\cup B)-\{r\}$ around $r$, such that $|R_1|=k|B_1|$ and $|R_2|=k|B_2|$, where $R_1$ $($resp. $B_1$$)$ is the set of red points $($resp. blue points$)$ of $(R\cup B)-\{r\}$ lying on or to the left of $\ell(r,b)$ and $R_2$ $($resp. $B_2$$)$ is the set of red points $($resp. blue points$)$ of $(R\cup B)-\{r\}$ lying on or to the right of $\ell(r,b')$.
\end{enumerate}
\end{lemma}
\begin{proof}
By a suitable rotation of the plane, we may assume that $r$ is the lowest point of $\CH{R\cup B}$, and $r_1$ (resp. $r_2$) is to the left (resp. right) of the vertical line passing through $r$. Note that $r_1$ is the first point and $r_2$ is the last point in the clockwise radial ordering of $(R\cup B)-\{r\}$ around $r$. See Figure~\ref{convex-3red-fig}. 
We define the function $f$ as follows: For every point $x$ in this radial ordering,  
\begin{align*}
f(x)= & k\cdot (\text{the number of points of $B$ lying on or to the left of $\ell(r,x)$})\\
& -(\text{the number of points of $R-\{r\}$ lying on or to the left of $\ell(r,x)$}).
\end{align*}
Based on this definition, we have $f(r_1)= -1$ and $f(r_2)=k|B|-(|R|-1)=0$. Along this radial ordering, the value of $f$ changes by $+k$ at every blue point and by $-1$ at every red point. Let $v$ be the point before $r_2$ in the radial ordering. Since $r_2$ decreases $f$ by $-1$ and $f(r_2)=0$, we have $f(v)=1$. Since $f(r_1)<0< f(v)$, there exists a point $b$ between $r_1$ and $v$ in the radial ordering such that $f(b)\geqslant 0$ and $f$ is negative at $b$'s predecessor. Let $b$ be the last such point between $r_1$ and $v$; it may happen that $b=v$. Observe that $b$ is blue and $f(b)< k$. 
We consider two cases, depending on whether $0<f(b)<k$ or $f(b)=0$.
\begin{itemize}
 \item $0<f(b)<k$. Since $b$ is blue and $r_2$ is red, there is at least one point after $b$ in the radial ordering. Define $R_1$, $B_1$, $R_2$ and $B_2$ as in the first statement of the lemma. See Figure~\ref{convex-3red-fig}(a). Let $\alpha=f(b)$. By definition of $f$, we have $\alpha =f(b)=k|B_1|-|R_1|$, and hence $|R_1|=k|B_1|-\alpha=k(|B_1|-1)+k-\alpha$. Moreover, 
\begin{align*}
|R_2| &= |R|-|R_1|-1\\
 &= (k|B|+1)-(k|B_1|-\alpha) -1\\
&= k(|B|-|B_1|)+\alpha\\
 &= k(|B_2|-1)+\alpha,
\end{align*}
where the last equality is valid because $b$ belongs to both $B_1$ and $B_2$.
\item $f(b)=0$. In the radial ordering there are at least $k+1$ points after $b$ since a red point only decreases the value of $f$ and $f(r_2)=0$. Let $b'$ be the successor of $b$ in the radial ordering. The point $b'$ is blue: If $b'$ is red, we have $f(b')=-1<0<f(v)$ and, thus, there exists a point $b''$ between $b'$ and $v$ such that $f(b'')\geqslant 0$ and $f$ is negative at the predecessor of $b''$, contradicting our choice of $b$. Define $R_1$, $B_1$, $R_2$ and $B_2$ as in the second statement of the lemma.  See Figure~\ref{convex-3red-fig}(b). Since $f(b)=k|B_1|-|R_1|=0$, we have $|R_1|=k|B_1|$. Moreover, 
\begin{align*}
|R_2| &= |R|-|R_1|-1\\
 &= (k|B|+1)-k|B_1|-1\\
&= k(|B|-|B_1|)\\
 &= k|B_2|,
\end{align*}
which completes the proof of the lemma.
\end{itemize}
\end{proof}

\subsection{Proof of Theorem~\ref{thr2}}

We use Lemma~\ref{convex-3blue} and Lemma~\ref{convex-3red} to prove Theorem~\ref{thr2}. Let $R$ and $B$ be two disjoint sets of points in the plane, such that $|R| = k|B|$, with
$k\geqslant 2$, and $R \cup B$ is in general position. We will present an algorithm, $\setup$, that constructs a plane bichromatic tree of maximum degree $k+1$ on $R\cup B$ such that each red vertex has degree at most 3. This algorithm uses two procedures, $\procA$ and $\procB$: 

\begin{table}[H]
\centering
\begin{tabularx}{0.95\textwidth}{|X|l}
\hline
{$\setup(R,B)$}
\\ \hline
\footnotesize{\bf Input:} A set $R$ of red points and a non-empty set $B$ of blue points, where $|R|=k|B|$, with $k\geqslant 2$, and $R\cup B$ is in general position.\\
\footnotesize{\bf Output:} A plane bichromatic $(k+1)$-tree on $R\cup B$ such that each red vertex has degree at most 3.
\\ \hline
\end{tabularx}
\end{table}
\vspace{-10pt}
\begin{table}[H]
\centering
\begin{tabularx}{0.95\textwidth}{|X|l}
\hline
{$\procA(R,B,b)$}
\\ \hline
\footnotesize{\bf Input:} A set $R$ of red points, a non-empty set $B$ of blue points, and a point $b\in B$, where $k(|B|-1)<|R|\leqslant k|B|$, with $k\geqslant 2$, and $b$ is on $\CH{R\cup B}$. \\
\footnotesize{\bf Output:} A plane bichromatic $(k+1)$-tree $T$ on $R\cup B$ where $\dgT{T}{b}=|R|-k(|B|-1)$ and each red vertex has degree at most 3.
\\ \hline
\end{tabularx}
\end{table}
\vspace{-10pt}
\begin{table}[H]
\centering
\begin{tabularx}{0.95\textwidth}{|X|l}
\hline
{$\procB(R,B,r)$}
\\ \hline
\footnotesize{\bf Input:} A set $R$ of red points, a non-empty set $B$ of blue points, and a point $r\in R$, where $|R|= k|B|+1$, with $k\geqslant 2$, and $r$ is on $\CH{R\cup B}$. \\
\footnotesize{\bf Output:} A plane bichromatic $(k+1)$-tree $T$ on $R\cup B$ where $\dgT{T}{r}\in\{1,2\}$ and each other red vertex has degree at most 3.
\\ \hline
\end{tabularx}
\end{table}

First we describe each of the procedures $\procA$ and $\procB$. Then we describe algorithm $\setup$. The procedures $\procA$ and $\procB$ will call each other. As we will see in the description of these procedures, when $\procA$ or $\procB$ is called recursively, the call is always on a smaller point set.
We now describe the base cases for $\procA$ and $\procB$. 

The base case for $\procA$ happens when $|B|=1$, i.e., $B=\{b\}$. In this case, we have $1\leqslant|R|\leqslant k$ and $2\leqslant|R\cup B|\leqslant k+1$. We connect all points of $R$ to $b$, and return the resulting star as a desired tree $T$ where $\dgT{T}{b}=|R|$ and each red vertex has degree 1.

The base case for $\procB$ happens when $|B|=1$; let $b$ be the only point in $B$. In this case, we have $|R|=k+1$ and $|R\cup B|=k+2$. We connect all points of $R$ to $b$, and return the resulting star as a desired tree $T$ where $\dgT{T}{b}=k+1$ and each red vertex has degree 1.

In Section~\ref{procA-section} we describe $\procA(R,B,b)$, whereas $\procB(R,B,r)$ will be described in Section~\ref{procB-section}. In these two sections, we assume that both $\procA$ and $\procB$ are correct for smaller point sets. 
\subsubsection{Procedure \procA}
\label{procA-section}
The procedure $\procA(R,B,b)$ takes as input a set $R$ of red points, a set $B$ of blue points, and a point $b\in B$, where $|B|\geqslant 2$, $k(|B| - 1) < |R| \leqslant k|B|$, with $k\geqslant 2$, and $b$ is on $\CH{R\cup B}$. Let $\alpha=|R|-k(|B|-1)$, and notice that $1\leqslant \alpha\leqslant k$. This procedure computes a plane bichromatic $(k+1)$-tree $T$ on $R\cup B$ where $\dgT{T}{b}=\alpha$ and each red vertex has degree at most 3. 

We consider two cases, depending on whether or not both vertices of $\CH{R\cup B}$ adjacent to $b$ belong to $B$.

\begin{enumerate}[wide, labelindent=0pt,label={\bf Case \arabic*:}]
\item Both vertices of $\CH{R\cup B}$ adjacent to $b$ belong to $B$. We apply Lemma~\ref{convex-3blue} on $R$, $B$, and $b$. Consider the $\alpha$ consecutive red points, $r_1,\dots,r_\alpha$, and the sets $R_1$, $R_2$, $B_1$, and $B_2$ in the statement of Lemma~\ref{convex-3blue}. Note that $r_1$ is a red point on $\CH{R_1\cup B_1}$ and $r_\alpha$ is a red point on $\CH{R_2\cup B_2}$. We distinguish between two cases:
 $1<\alpha\leqslant k$ and $\alpha=1$.
\begin{enumerate}[leftmargin=*,itemindent=30pt, labelindent=0pt,label=Case 1.\arabic*:]
\item $1<\alpha\leqslant k$. In this case $r_1\neq r_\alpha$. Moreover $\CH{R_1\cup B_1}$ and $\CH{R_2 \cup B_2}$ are disjoint. Let $T_1$ be the plane bichromatic $(k+1)$-tree obtained by running $\procB$ on $R_1$, $B_1$, $r_1$; note that $\dgT{T_1}{r_1}\in\{1,2\}$, all other red points in $T_1$ have degree at most 3, and $|R_1\cup B_1|<|R\cup B|$. Similarly, let $T_2$ be the plane bichromatic $(k+1)$-tree obtained by running $\procB$ on $R_2$, $B_2$, $r_\alpha$; note that $\dgT{T_2}{r_\alpha}\in\{1,2\}$, all other red points in $T_2$ have degree at most 3,  and $|R_2\cup B_2|<|R\cup B|$. Let $S$ be the star obtained by connecting the vertices $r_1,\dots, r_\alpha$ to $b$. Then, we obtain a desired tree $T=T_1\cup T_2\cup S$. See Figure~\ref{convex-3blue-fig}. $T$ is a plane bichromatic $(k+1)$-tree on $R\cup B$ with $\dgT{T}{r_1}\in\{2,3\}$, $\dgT{T}{r_\alpha}\in\{2,3\}$, $\dgT{T}{b}=\alpha$, and $\dgT{T}{r_i}=1$ where $1< i<\alpha$. 

\item $\alpha=1$. In this case $r_1=r_\alpha$ and $|R|=k(|B|-1)+1$. Moreover, $r_1\in R_1\cap R_2$. If we handle this case as in the previous case, then it is possible for $r_1$ to be incident on two edges in each of $T_1$ and $T_2$, and incident on one edge in $S$. This makes $\dgT{T}{r_1}=5$. If $k\geqslant 4$, then $T$ is a desired $(k+1)$-tree. But, if $k=2,3$, then $T$ would not be a $(k+1)$-tree. Thus, we handle the case when $\alpha=1$ differently. 

Let $x_1$ and $y_1$ be the two blue neighbors of $b$ on $\CH{R\cup B}$. By a suitable rotation of the plane, we may assume that $b$ is the lowest point of $\CH{R\cup B}$, and $x_1$ (resp. $y_1$) is to the left (resp. right) of the vertical line
passing through $b$. 
Let $C_1=(x_1,\dots, x_j=r_1)$ be the sequence of points on the boundary of $\CH{R_1\cup B_1}$ from $x_1$ to $r_1$ that are visible from $b$. Similarly, define $C_2=(y_1,\dots, r_1)$ on $\CH{R_2\cup B_2}$. See Figure~\ref{convex-3blue-1red-fig}.
Let $x_s$ be the first red point in the sequence $C_1$, and let $y_t$ be the first red point in the sequence $C_2$. Note that $s, t\geqslant 2$. It is possible for $x_s$ or $y_t$ or both to be $r_1$. Consider the subsequences $C'_1=(x_1,\dots,x_s)$ and $C''_1=(x_s,\dots,r_1)$ of $C_1$ as depicted in Figure~\ref{convex-3blue-1red-fig}(a). Similarly, consider the subsequences $C'_2=(y_1,\dots, y_t)$ and $C''_2=(y_t,\dots, r_1)$ of $C_2$. Let $l_1$ and $l_2$ be the lines passing through $(x_{s-1},x_s)$ and $(y_{t-1},y_t)$, respectively. Note that $l_1$ is tangent to $\CH{R_1\cup B_1}$ and $l_2$ is tangent to $\CH{R_2\cup B_2}$. 

\begin{figure}[H]
  \centering
\setlength{\tabcolsep}{0in}
  $\begin{tabular}{cc}
 \multicolumn{1}{m{.5\columnwidth}}{\centering\includegraphics[width=.45\columnwidth]{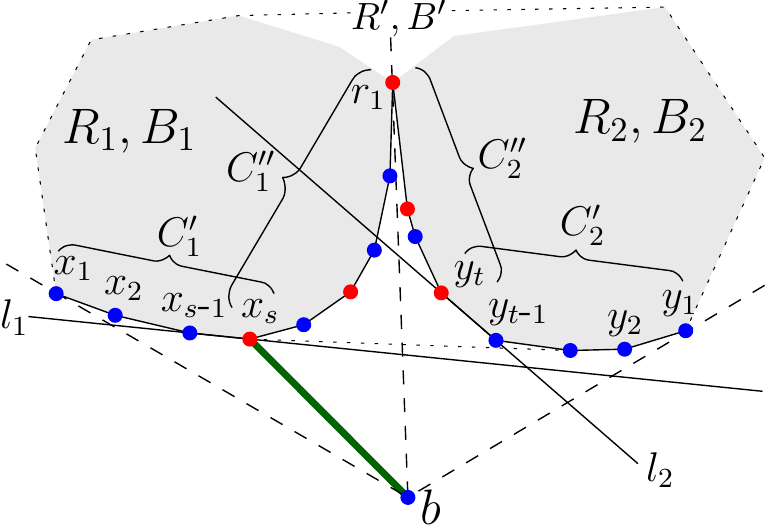}}
&\multicolumn{1}{m{.5\columnwidth}}{\centering\includegraphics[width=.45\columnwidth]{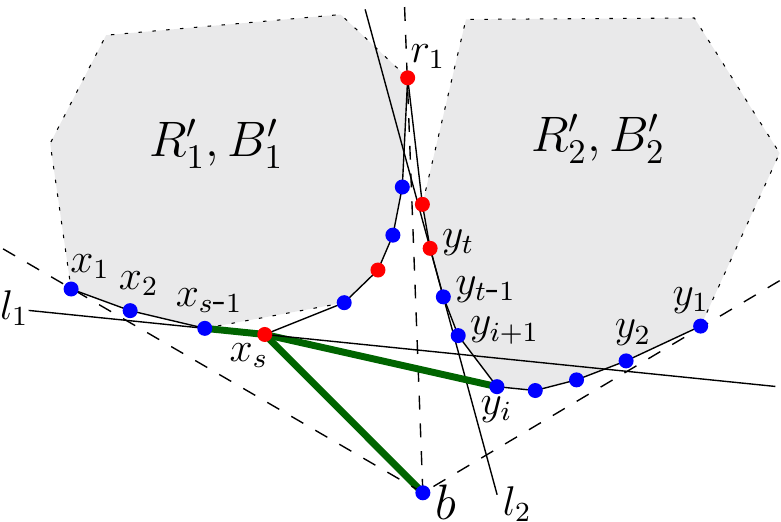}}\\
(a) & (b)
\end{tabular}$
  \caption{(a) $l_1$ does not intersect the interior of $\CH{R\cup B}$, and (b) $l_1$ intersects $C'_2$.}
\label{convex-3blue-1red-fig}
\end{figure}

We consider two cases, depending on whether or not $l_1$ intersects $C_2$ and $l_2$ intersects $C_1$.
\begin{enumerate}[leftmargin=*,itemindent=40pt, labelindent=0pt,label=Case 1.2.\arabic*:]
 \item {\em $l_1$ does not intersect $C_2$, or $l_2$ does not intersect $C_1$.} Because of symmetry, we assume that $l_1$ does not intersect $C_2$. Note that in this case $l_1$ does not intersect the interior of $\CH{R\cup B-\{b\}}$. Let $R'=R$ and $B'=B-\{b\}$; note that $|R'|=|R|=k(|B|-1)+1=k|B'|+1$. In addition, $x_s$ is on $\CH{R'\cup B'}$. See Figure~\ref{convex-3blue-1red-fig}(a). Let $T'$ be the plane bichromatic ($k+1$)-tree obtained by $\procB(R', B', x_s)$. Note that $\dgT{T'}{x_s}\in\{1,2\}$, all other red points in $T'$ have degree at most 3, and $|R'\cup B'|<|R\cup B|$. We obtain a desired tree $T=T'\cup \{(b,x_s)\}$. $T$ is a plane bichromatic$(k+1)$-tree on $R\cup B$ with $\dgT{T}{b}=\alpha=1$ and $\dgT{T}{x_s}\in\{2,3\}$. 

\item {\em $l_1$ intersects $C_2$, and $l_2$ intersects $C_1$.} We distinguish between two cases:

\begin{enumerate}[leftmargin=*,itemindent=50pt, labelindent=0pt,label=Case 1.2.2.\arabic*:]

\item {\em $l_1$ intersects $C'_2$, or $l_2$ intersects $C'_1$.} Because of symmetry, we assume that $l_1$ intersects $C'_2$. Let $(y_i,y_{i+1})$, with $1\leqslant i < t$, be the leftmost edge of $C'_2$ that is intersected by $l_1$ (note that $l_1$ may intersect two edges of $C'_2$). Observe that $y_i$ is a blue point. Let $R'_1=R_1-\{x_s\}$, $B'_1=B_1$, $R'_2=R_2-\{r_1\}$ and $B'_2=B_2$ as shown in Figure~\ref{convex-3blue-1red-fig}(b). Note that $|R'_1|=k|B'_1|$ and $|R'_2|=k|B'_2|$. In addition, $\CH{R'_1\cup B'_1}$ and $\CH{R'_2\cup B'_2}$ are disjoint, $x_{s-1}$ is a blue point on $\CH{R'_1\cup B'_1}$, and $y_i$ is a blue point on $\CH{R'_2\cup B'_2}$. Let $T_1$ be the plane bichromatic $(k+1)$-tree obtained by the recursive call $\procA(R'_1, B'_1, x_{s-1})$, and let $T_2$ be the plane bichromatic $(k+1)$-tree obtained by the recursive call $\procA(R'_2, B'_2, y_i)$. Note that $\dgT{T_1}{x_{s-1}}=k$, $\dgT{T_2}{y_i}=k$, all red points in $T_1$ and $T_2$ have degree at most 3, $|R'_1\cup B'_1|<|R\cup B|$, and $|R'_2\cup B'_2|<|R\cup B|$. We obtain a desired tree $T=T_1\cup T_2\cup \{(b,x_s), (x_{s-1},x_s), (y_i,x_s)\}$; see Figure~\ref{convex-3blue-1red-fig}(b). $T$ is a plane bichromatic $(k+1)$-tree on $R\cup B$ with $\dgT{T}{b}=\alpha=1$, $\dgT{T}{x_s}=3$, $\dgT{T}{x_{s-1}}=k+1$ and $\dgT{T}{y_i}=k+1$. 

\item {\em $l_1$ intersects $C''_2$, and $l_2$ intersects $C''_1$.} In this case $Q=(x_{s-1},x_s,\allowbreak y_t,\allowbreak y_{t-1})$ is a convex quadrilateral because $l_1\cap (y_{t-1},y_t)=\emptyset$ and $l_2\cap (x_{s-1},x_s)=\emptyset$. Moreover, $Q$ does not have any point of $R\cup B$ in its interior and it has no intersection with the interiors of $\CH{R_1\cup B_1}$ and $\CH{R_2\cup B_2}$. We handle this case as in Case 1.2.2.1 with the blue point $y_{t-1}$ playing the role of $y_i$. Observe that this construction gives a valid tree even if $x_s=y_t=r_1$.
\end{enumerate}
\end{enumerate}
\end{enumerate}

\item At least one of the vertices on $\CH{R\cup B}$ adjacent to $b$ does not belong to $B$. Let $x_1$ be such a vertex that belongs to $R$. Initialize $X=\{x_1\}$. If at least one of the vertices of $\CH{(R-X)\cup B}$ adjacent to $b$ does not belong to $b$, let $x_2$ be such a red point. Add $x_2$ to the set $X$. Repeat this process on $\CH{(R-X)\cup B}$ until $|X|=\alpha$ or both neighbors of $b$ on $\CH{(R-X)\cup B}$ are blue points. Let $x_1,\dots,x_\beta$ be the sequence of red points added to $X$ in this process. After this process we have $|X|=\beta$, where $1\leqslant \beta\leqslant \alpha$. Let $S_1$ be the star obtained by connecting all points of $X$ to $b$. See Figure~\ref{red-blue-fig} where $S_1$ is shown with green bold edges. Observe that $\dgT{S_1}{b}=\beta$. We distinguish between two cases: $\beta=\alpha$ and $1\leqslant \beta<\alpha$. 

\begin{figure}[htb]
  \centering
\setlength{\tabcolsep}{0in}
  $\begin{tabular}{cc}
 \multicolumn{1}{m{.5\columnwidth}}{\centering\includegraphics[width=.3\columnwidth]{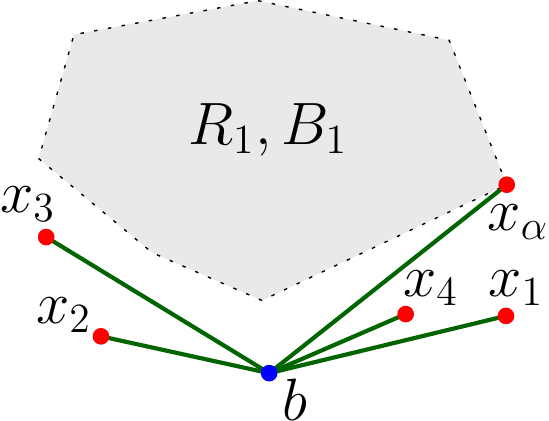}}
&\multicolumn{1}{m{.5\columnwidth}}{\centering\includegraphics[width=.3\columnwidth]{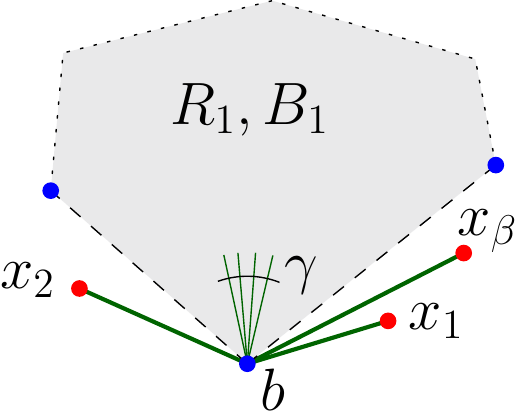}}\\
(a) & (b)
\end{tabular}$
  \caption{The edges in $S_1$ are in bold where (a) $\beta=\alpha$, and (b) $1\leqslant \beta<\alpha$.}
\label{red-blue-fig}
\end{figure}

\begin{enumerate}[leftmargin=*,itemindent=30pt, labelindent=0pt,label=Case 2.\arabic*:]
 \item $\beta=\alpha$. Let $R_1=(R-X)\cup\{x_\alpha\}$ and $B_1=B-\{b\}$. See Figure~\ref{red-blue-fig}(a). Note that $x_\alpha$ is a red point on $\CH{R_1\cup B_1}$ and  
\begin{align*}
|R_1| &= |R|-\alpha +1\\
 &= |R|-(|R|-k(|B|-1))+1\\
&= k|B_1|+1.
\end{align*}
Let $T_1$ be the plane bichromatic $(k+1)$-tree obtained by $\procB(R_1, B_1, x_\alpha)$ with $\dgT{T_1}{x_\alpha}\in\{1,2\}$ and all other red points in $T_1$ are of degree at most 3. Note that $|R_1\cup B_1|<|R\cup B|$. We obtain a desired tree $T=T_1\cup S_1$. $T$ is a plane bichromatic $(k+1)$-tree on $R\cup B$ with $\dgT{T}{x_\alpha}\in\{2,3\}$ and $\dgT{T}{b}=\alpha=|R|-k(|B|-1)$ as required.

\item $1\leqslant\beta<\alpha$. In this case both vertices of $\CH{(R-X) \cup B}$ adjacent to $b$ are blue points. Let $R_1=R-X$ and $B_1=B$. See Figure~\ref{red-blue-fig}(b). Let $\gamma=\alpha - \beta$ and note that $1\leqslant \gamma< \alpha\leqslant k$. Then,  
\begin{align*}
|R_1| &= |R|-\beta\\
   &= (k(|B|-1)+\alpha)-\beta\\
   &= k(|B_1|-1)+\gamma.
\end{align*}
Thus, $k(|B_1|-1)<|R_1|\leqslant k|B_1|$. Let $T_1$ be the plane bichromatic $(k+1)$-tree obtained by the recursive call $\procA(R_1, B_1, b)$ with $\dgT{T_1}{b}=\gamma$ and all red points of $T_1$ are of degree at most 3. Note that $|R_1\cup B_1|<|R\cup B|$. We obtain a desired tree $T=T_1\cup S_1$. $T$ is a plane bichromatic $(k+1)$-tree on $R\cup B$ with $\dgT{T}{b}=\beta+\gamma=\alpha$.
\end{enumerate}
\end{enumerate}
\subsubsection{Procedure \procB}
\label{procB-section}
The procedure $\procB(R,B,r)$ takes as input a set $R$ of red points, a set $B$ of blue points, and a point $r\in B$, where $|B|\geqslant 2$, $|R| = k|B| + 1$, with $k\geqslant 2$, and $r$ is on $\CH{R\cup B}$. This procedure computes a plane bichromatic $(k+1)$-tree $T$ on $R\cup B$ where $\dgT{T}{r}\in\{1,2\}$ and each other red vertex has degree at most 3.
 
We consider two cases, depending on whether or not both vertices of $\CH{R\cup B}$ adjacent to $r$ belong to $R$.
\begin{enumerate}[wide, labelindent=0pt,label={\bf Case \arabic*:}]
\item At least one of the vertices on $\CH{R\cup B}$ adjacent to $r$ does not belong to $R$. Let $b$ be such a point belonging to $B$.
Let $R_1=R-\{r\}$, $B_1=B$. Note that $|R_1|=k|B_1|$, and $b$ is on $\CH{R_1\cup B_1}$. Let $T_1$ be the plane bichromatic $(k+1)$-tree obtained by $\procA(R_1, B_1, b)$ with $\dgT{T_1}{b}=k$ and all red points of $T_1$ are of degree at most 3. Note that $|R_1\cup B_1|<|R\cup B|$. Then, we obtain a desired tree $T =T_1\cup\{(r,b)\}$. $T$ is a plane bichromatic $(k + 1)$-tree on $R\cup B$ with $\dgT{T}{b}=k+1$ and $\dgT{T}{r}=1$.

\item Both vertices of $\CH{R\cup B}$ adjacent to $r$ belong to $R$. In this case, by Lemma~\ref{convex-3red} there are two possibilities:

\begin{enumerate}[leftmargin=*,itemindent=30pt, labelindent=0pt,label=Case 2.\arabic*:]
 \item The first statement in Lemma~\ref{convex-3red} holds. Consider the blue point $b$ and the sets $R_1$, $R_2$, $B_1$, and $B_2$ in this statement. Note that $b$ is a blue point on $\CH{R_1\cup B_1}$ and on $\CH{R_2\cup B_2}$. Let $T_1$ and $T_2$ be the plane bichromatic $(k+1)$-trees obtained by running $\procA(R_1, B_1, b)$ and $\procA(R_2, B_2, b)$, respectively. Note that $\dgT{T_1}{b}= k-\alpha$, $\dgT{T_2}{b}= \alpha$, all red points of $T_1$ and $T_2$ have degree at most 3, $|R_1\cup B_1|<|R\cup B|$, and $|R_2\cup B_2|<|R\cup B|$. We obtain a desired tree $T=T_1\cup T_2\cup \{(r,b)\}$. See Figure~\ref{convex-3red-fig}(a). $T$ is a plane bichromatic $(k+1)$-tree on $R\cup B$ with $\dgT{T}{r}=1$ and $\dgT{T}{b}=(k-\alpha)+\alpha+1=k+1$.
\item The second statement in Lemma~\ref{convex-3red} holds. Consider the blue points $b, b'$ and the sets $R_1$, $R_2$, $B_1$, and $B_2$ in this statement. Note that $b$ is a blue point on $\CH{R_1\cup B_1}$ and $b'$ is a blue point on $\CH{R_2\cup B_2}$. Let $T_1$ and $T_2$ be the plane bichromatic $(k+1)$-trees obtained by $\procA(R_1, B_1, b)$ and $\procA(R_2, B_2, b')$. Note that $\dgT{T_1}{b}= k$, $\dgT{T_2}{b'}= k$, all red points of $T_1$ and $T_2$ have degree at most 3, $|R_1\cup B_1|<|R\cup B|$, and $|R_2\cup B_2|<|R\cup B|$. We obtain a desired tree $T=T_1\cup T_2\cup \{(r,b),(r,b')\}$. See Figure~\ref{convex-3red-fig}(b). $T$ is a plane bichromatic $(k+1)$-tree on $R\cup B$ with $\dgT{T}{b}=k+1$, $\dgT{T}{b'}=k+1$, and $\dgT{T}{r}=2$.
\end{enumerate}
\end{enumerate}
\subsubsection{Algorithm \setup}
Algorithm $\setup(R,B)$ takes as input a set $R$ of red points and a non-empty set $B$ of blue points, where $|R|=k|B|$, with $k\geqslant 2$, and $R\cup B$ is in general position. This algorithm constructs a plane bichromatic $(k+1)$-tree $T$ on $R\cup B$ such that each red vertex has degree at most 3. By Observation~\ref{obs1}, $T$ has only one blue vertex of degree $k$ and the other blue vertices are of degree $k+1$.
We consider two cases, depending on whether or not all vertices of $\CH{R\cup B}$ belong to $R$.
\vspace{5pt}

\noindent{\bf Case 1:} At least one of the vertices of $\CH{R\cup B}$ belongs to $B$. Let $b$ be such a vertex. Let $T$ be the tree obtained by running $\procA(R, B,b)$. $T$ is a plane bichromatic $(k+1)$-tree on $R\cup B$ with $\dgT{T}{b}=k$ and all red vertices of $T$ are of degree at most 3. Notice that $b$ is the only blue vertex of degree $k$ in $T$.

\vspace{5pt}

\begin{wrapfigure}{r}{0.35\textwidth}
  \begin{center}
\vspace{-30pt}
\includegraphics[width=.33\textwidth]{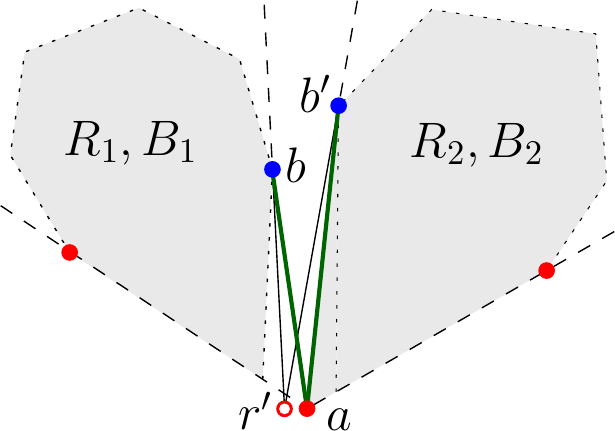}
\vspace{-10pt}
\caption{A dummy red point $r'$ is placed very close to $a$.}
  \label{dummy-red-fig}
  \end{center}
\vspace{-20pt}
\end{wrapfigure}
\noindent{\bf Case 2:} All vertices of $\CH{R\cup B}$ belong to $R$. Let $a$ be an arbitrary red point on $\CH{R\cup B}$. By a suitable rotation of the plane, we may assume that $a$ is the lowest point of $\CH{R\cup B}$. We add a dummy red point $r'$ at a sufficiently small distance $\epsilon$ to the left of $a$ such that the radial ordering of the points in $(R\cup B) -\{a\}$ around $r'$ is the same as their radial ordering around $a$. See Figure~\ref{dummy-red-fig}. Now we consider the radial ordering of the points in $R\cup B$ (including $a$) around $r'$. 
We apply Lemma~\ref{convex-3red} with $r'$ playing the role of $r$. There are two possibilities:

\begin{enumerate}[leftmargin=*,itemindent=25pt, labelindent=0pt,label=Case 2.\arabic*:] 
\item The first statement in Lemma~\ref{convex-3red} holds. Consider the blue point $b$ and the sets $R_1$, $R_2$, $B_1$, and $B_2$ as in the first statement of Lemma~\ref{convex-3red}. Note that $a\in R_2$, $r'\notin R_1\cup R_2$, and $b$ is a blue point on $\CH{R_1\cup B_1}$ and on $\CH{R_2\cup B_2}$. Let $T_1$ and $T_2$ be the plane bichromatic $(k+1)$-trees obtained by $\procA(R_1, B_1, b)$ and $\procA(R_2, B_2, b)$, respectively. Note that $\dgT{T_1}{b}= k-\alpha$, $\dgT{T_2}{b}= \alpha$, and all red vertices of $T_1$ and $T_2$ have degree at most 3. We obtain a desired tree $T=T_1\cup T_2$ with $\dgT{T}{b}=k-\alpha+\alpha=k$; $b$ is the only blue vertex of degree $k$ in $T$.
\item The second statement in Lemma~\ref{convex-3red} holds. Consider the blue points $b, b'$ and the sets $R_1$, $R_2$, $B_1$, and $B_2$ as in the second statement of Lemma~\ref{convex-3red}. Note that $a\in R_2$, $r'\notin R_1\cup R_2$, $b$ is a blue point on $\CH{R_1\cup B_1}$, and $b'$ is a blue point on $\CH{R_2\cup B_2}$. 
If we compute trees on $R_1\cup B_1$ and $R_2\cup B_2$ and discard $r'$, as we did in the previous case, then the resulting graph is not connected and hence it is not a tree. Thus, we handle this case in a different way. First we remove $a$ from $R_2$ as shown in Figure~\ref{dummy-red-fig}; this makes $|R_2|=k|B_2|-1=k(|B_2|-1)+(k-1)$. Note that $\CH{R_1\cup B_1}$ and $\CH{R_2\cup B_2}$ are disjoint. Let $T_1$ and $T_2$ be the plane bichromatic $(k+1)$-trees obtained by $\procA(R_1, B_1, b)$ and $\procA(R_2, B_2, b')$, respectively. Note that $\dgT{T_1}{b}=k$,  $\dgT{T_2}{b'}=k-1$, and all red vertices of $T_1$ and $T_2$ have degree at most 3. We obtain a desired tree $T=T_1\cup T_2\cup \{(a,b),(a,b')\}$ with $\dgT{T}{a}=2$, $\dgT{T}{b}=k+1$ and $\dgT{T}{b'}=k$; $b'$ is the only vertex of degree $k$ in $T$.
\end{enumerate}
This concludes the description of algorithm $\setup$. The pseudo code for $\procA$, $\procB$, and $\setup$ are given in Algorithms~\ref{procA},~\ref{procB}, and~\ref{setup}, respectively. 
\begin{algorithm}[H]                 
\caption{\procA$(R, B, b)$}          
\label{procA} 
\require{$R$, $B$, and $b\in B$ such that $k(|B|-1)<|R|\leqslant k|B|$, with $k\geqslant 2$, and $b$ on $\CH{R\cup B}$.}\\
\ensure{a plane bichromatic $(k+1)$-tree $T$ on $R\cup B$ with $\dgT{T}{b}=|R|-k(|B|-1)$ and each red vertex has degree at most 3.}
\begin{algorithmic}[1]
    \State $S_1, X\gets\emptyset;~\alpha\gets|R|-k(|B|-1)$
    \If {both neighbors of $b$ on $\CH{R\cup B}$ are blue}
	  \State $r_1,\dots, r_\alpha \gets$ radially consecutive red points obtained in Lemma~\ref{convex-3blue}	  \State $R_1,B_1,R_2,B_2 \gets$ point sets obtained in Lemma~\ref{convex-3blue}

	  \If {$\alpha=1$}
	      $T\gets$ the tree obtained in Case 1.2
	  \Else
	      ~$T\gets\{(b,r_1),\dots,(b,r_\alpha)\}\cup\procB(R_1,B_1,r_1)\cup \procB(R_2,B_2,r_\alpha)$
	  \EndIf
    \Else
    \While {$b$ has a red neighbor on $\CH{(R- X)\cup B}$ and $\dgT{S_1}{b}< \alpha$}
	  \State $x\gets$ a red neighbor of $b$ on $\CH{(R-X)\cup B}$
	  \State $X\gets X\cup \{x\},~S_1\gets S_1\cup \{(b,x)\}$
     \EndWhile
     \If {$\dgT{S_1}{b}=\alpha$}
	  $T\gets S_1\cup \procB(R\cup\{x\},B-\{b\},x)$
      \Else
	  ~$T\gets S_1\cup \procA(R-X,B,b)$
    \EndIf
\EndIf
\State \Return $T$
\end{algorithmic}
\end{algorithm}

\begin{algorithm}[H]                 
\caption{\procB$(R, B, r)$}          
\label{procB} 
\require{$R$, $B$, and $r\in R$ such that $|R|=k|B|+1$, with $k\geqslant 2$, and $r$ on $\CH{R\cup B}$.}\\
\ensure{a plane bichromatic $(k+1)$-tree $T$ on $R\cup B$ with $\dgT{T}{r}\in\{1,2\}$ and each other red vertex has degree at most 3.}
\begin{algorithmic}[1]
    \If {$B=\emptyset$}
	\Return $\emptyset$
    \EndIf
    \If {$r$ has any blue neighbor on $\CH{R\cup B}$}
	  \State {$b\gets$ a blue neighbor of $r$ on $\CH{R\cup B}$}
	  \State $T\gets (r,b)\cup\procA(R-\{r\},B, b)$
      \Else
     \If {case 1 of Lemma~\ref{convex-3red} with respect to $r$}
	  \State $b\gets$ the point obtained in case 1 of Lemma~\ref{convex-3red}
	  \State $R_1, B_1, R_2, B_2\gets$ point sets obtained in case 1 of Lemma~\ref{convex-3red}
	  \State $T\gets \{(r,b)\}\cup\procA(R_1,B_1,b)\cup \procA(R_2,B_2,b)$
     \Else
	  \State $b, b'\gets$ the points obtained in case 2 of Lemma~\ref{convex-3red}
	  \State $R_1, B_1, R_2, B_2\gets$ point sets obtained in case 2 of Lemma~\ref{convex-3red}
	  \State $T\gets \{(r,b),(r,b')\}\cup\procA(R_1,B_1,b)\cup \procA(R_2,B_2,b')$
     \EndIf
\EndIf
\State \Return $T$
\end{algorithmic}
\end{algorithm}

\begin{algorithm}[H]                 
\caption{\setup$(R, B)$}          
\label{setup} 
\require{a set $R$ of red points, a set $B$ of blue points such that $|R|=k|B|$, with $k\geqslant 2$.}\\
\ensure{a plane bichromatic $(k+1)$-tree $T$ on $R\cup B$ such that each red vertex has degree at most 3.}
\begin{algorithmic}[1]
    \If {there is a blue point $b$ on $\CH{R\cup B}$}
	  \State $T\gets\procA(R,B, b)$
      \Else
    \State $a\gets$ a red point on $\CH{R\cup B}$
    \State $r'\gets$ a dummy red point very close and to the left of $a$ on $\CH{R\cup B}$
     \If {case 1 of Lemma~\ref{convex-3red} with respect to $r'$}
	  \State $b\gets$ the point obtained in case 1 of Lemma~\ref{convex-3red}
	  \State $R_1, B_1, R_2, B_2\gets$ point sets obtained in case 1 of Lemma~\ref{convex-3red}
	  \State $T\gets \procA(R_1,B_1,b)\cup \procA(R_2,B_2,b)$
     \Else
	  \State $b, b'\gets$ the points obtained in case 2 of Lemma~\ref{convex-3red}
	  \State $R_1, B_1, R_2, B_2\gets$ point sets obtained in case 2 of Lemma~\ref{convex-3red}
	  \State $T\gets \{(a,b),(a,b')\}\cup\procA(R_1,B_1,b)\cup \procA(R_2-\{a\},B_2,b')$
     \EndIf
\EndIf
\State \Return $T$
\end{algorithmic}
\end{algorithm}

\section{Proof of Conjecture~\ref{conj2}}
\label{general-section}
In this section we prove Conjecture~\ref{conj2}: Given two disjoint sets, $R$ and $B$, of points in the plane such that $|B|\leqslant|R|$, $\left\lceil\frac{|R|}{|B|}\right\rceil=k$, with $k\geqslant 2$, and $R\cup B$ is in general position, we show how to compute a plane bichromatic tree on $R\cup B$ whose maximum vertex degree is at most $k+1$.  

We prove Conjecture~\ref{conj2} by modifying the algorithm presented in Section~\ref{k+1-section}.
Since $|B|\leqslant|R|$ and $\left\lceil\frac{|R|}{|B|}\right\rceil=k$, we have $(k-1)|B|<|R|\leqslant k|B|$, with $k\geqslant 2$. If $|R|=k|B|$, then by Theorem~\ref{thr2} there exists a plane bichromatic $(k+1)$-tree on $R\cup B$. Therefore, assume that $(k-1)|B|<|R|< k|B|$. Let $\omega = k|B|-|R|$. Observe that $1\leqslant \omega< |B|<|R|$. Our main idea for proving Conjecture~\ref{conj2} is to add $\omega$ dummy red points and then use (a modified version of) the algorithm presented in Section~\ref{k+1-section} to obtain a $(k+1)$-tree on $R\cup B$. 

We pick an arbitrary subset $S$ of $R$ of size $\omega$. For each point $r\in S$ we add a new red point $r'$ at a sufficiently small distance $\epsilon$ to $r$ such that the cyclic orders of $(R\cup B)-\{r\}$ around $r$ and around $r'$ are equal. Let $R'$ be the set of these new points. 

We call the points in $R'$ as {\em dummy} red points, the points in $S$ as {\em saturated} red points, and the points in $R-S$ as {\em unsaturated} red points. Each dummy red point corresponds to a saturated red point, and vice versa. 

In Section~\ref{unify-section} we present a method that computes a plane bichromatic $(k+1)$-tree on $R\cup B$ for the case when $k\geqslant 4$. In Section~\ref{pass-refine-section} we present another method that computes a plane bichromatic $(k+1)$-tree on $R\cup B$ for the case when $k\geqslant 2$, which proves Conjecture~\ref{conj2}.

\begin{figure}[htb]
  \centering
  \includegraphics[width=.55\columnwidth]{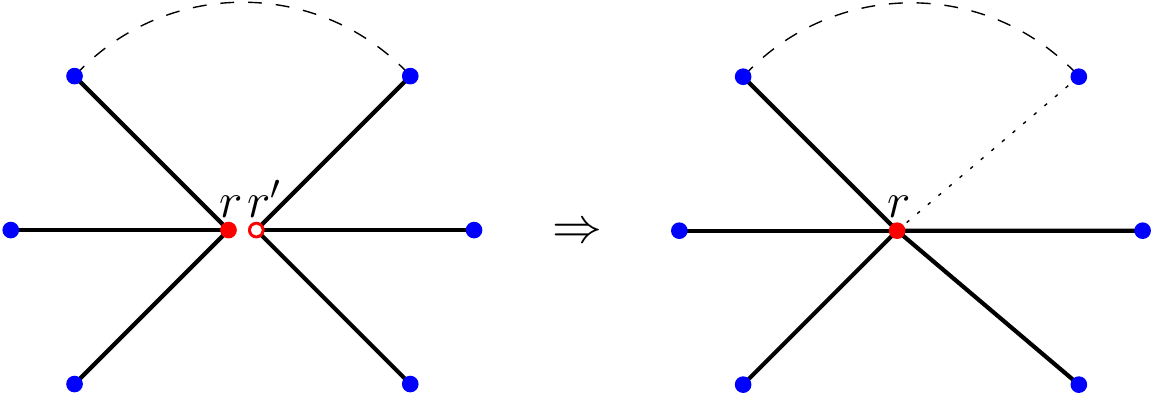}
 \caption{Unifying the dummy red point $r'$ with its corresponding saturated red point $r$.}
  \label{identify-fig}
\end{figure}

\subsection{Unify Method}
\label{unify-section}
In this section we present a method that computes a plane bichromatic tree on $R\cup B$ whose maximum vertex degree is at most $\max\{5, k+1\}$.

Note that $|R\cup R'|$ is equal to $k|B|$. 
Let $T'$ be the plane bichromatic $(k+1)$-tree obtained by running algorithm $\setup$ on $R\cup R'$, $B$. Consider any dummy red vertex $r'$ in $T'$. Let $r$ be the saturated red vertex in $T'$ that corresponds to $r'$. Recall that $\dgT{T'}{r}\leqslant 3$ and $\dgT{T'}{r'}\leqslant 3$. We {\em unify} $r'$ with $r$ as follows. Connect all points adjacent to $r'$ to $r$, and then remove $r'$ from $R$. This creates a cycle in the resulting graph where $r$ is on that cycle. Then, remove one of the edges of the cycle incident on $r$; see Figure~\ref{identify-fig}. As a result, we obtain a tree $T''$ with $\dgT{T''}{r}\leqslant 5$. Let $T$ be the tree obtained after unifying all the dummy red points with their corresponding saturated red points. Then, $T$ is a plane bichromatic tree on $R\cup B$ whose maximum vertex degree is at most $\max\{5, k+1\}$. If $k\geqslant4$, then this construction gives a plane bichromatic $(k+1)$-tree. However, if $k\in\{2,3\}$ then this construction may give a tree of degree 5. 

\subsection{Pair-Refine Method}
\label{pass-refine-section}
In this section we present a method that computes a plane bichromatic tree on $R\cup B$ whose maximum vertex degree is at most $k+1$ when $k\in\{2,3\}$. However, this method works for all $k\geqslant 2$.

Recall that $|R\cup R'|$ is equal to $k|B|$. Consider a plane bichromatic $(k+1)$-tree $T'$ obtained by $\setup(R\cup R', B)$. If in $T'$ all the dummy points are leaves, i.e., $\dgT{T'}{r'}=1$ for all $r'\in R'$, then a desired tree can easily be obtained by removing the dummy vertices from $T'$. In this section we show how to extend/modify the algorithm of Section~\ref{k+1-section} to obtain a plane bichromatic $(k+1)$-tree $T'$ where all the dummy points are leaves, i.e., no dummy point appears as an internal node in $T'$. For simplicity, we write $R$ for $R\cup R'$. We denote the set of dummy points in a set $R$ by $\mathcal{D}_R$. A dummy red point $r'$ can appear as an internal node only in one of the following cases:

\begin{enumerate}[wide, labelindent=0pt,label={Case \arabic*:}]
\item $\procA$ calls $\procB$ on $r'$ when $\alpha\geqslant 2$
\item $r'=x_s$ in $\procA$ when $\alpha=1$
\item $\setup$ calls $\procB$ on $r'$.
\end{enumerate}

We show how to handle each of these cases.
Consider the moment we want to call $\procB$ on $r'$ either in $\procA$ or in $\setup$. Let $r$ be the saturated red point corresponding to $r'$. In order to prevent $r'$ from becoming an internal node, we either delete $r'$ or pair $r'$ to an unsaturated red point different from $r$; this makes $r$ unsaturated. Then, we call $\procB$ on $r$ (instead of on $r'$). To see why it is necessary to delete or pair $r'$, assume we called $\procB$ on $r$ while keeping $r'$. Later on, during the algorithm we may call $\procB$ on $r'$ for the second time. Now if we run $\procB$ on $r$ for the second time, it may lead to creating a cycle in $T'$ or increasing the degree of $r$. 

\begin{figure}[htb]
  \centering
\setlength{\tabcolsep}{0in}
  $\begin{tabular}{cc}
 \multicolumn{1}{m{.35\columnwidth}}{\centering\includegraphics[width=.2\columnwidth]{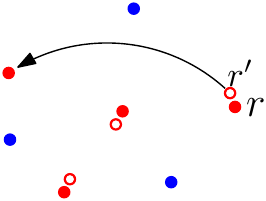}}
&\multicolumn{1}{m{.65\columnwidth}}{\centering\includegraphics[width=.5\columnwidth]{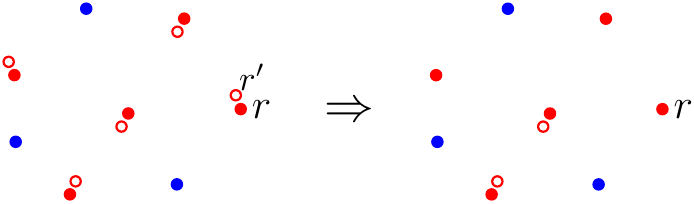}}\\
(a) & (b)
\end{tabular}$
  \caption{$\passrefine$ method: (a) moving the dummy red point $r'$ to an unsaturated red point where $|R|=2|B|+1$, (b) refining the set of dummy red points by removing $|B|$-many of them (including $r'$) where $|R|=3|B|+1$; after refinement we have $|R|=2|B|+1$. Disks represent blue points, big circles represent non-dummy red points, and small circles represent dummy red points}
\label{pass-refine-fig}
\end{figure}

Since we want to run $\procB$ on $r$ right after deleting or pairing $r'$, we have to make sure that the conditions of $\procB$ are satisfied after deleting or pairing $r'$. Thus, we define two new operations: {\em pair} and {\em refine}. 
The pair operation moves a dummy red point from its corresponding saturated red point to an unsaturated red point; see Figure~\ref{pass-refine-fig}(a). This makes the saturated red point to become unsaturated, and vice versa. The refine operation deletes some dummy red points from $R$; this makes their corresponding saturated points to become unsaturated; see Figure~\ref{pass-refine-fig}(b). 
We define $\passrefine$ method is depicted in Algorithm~\ref{pass-refine}. This method gets a set $R$ of red points (some of them are dummy), a non-empty set $B$ of blue points, and a saturated point $r\in R$ that is on $\CH{(R-\mathcal{D}_R)\cup B}$ as input. Since we call $\passrefine$ method right before calling $\procB$, we are in the case where $|R|=k|B|+1$, with $k\geqslant 2$. Let $r'$ be the dummy point in $R$ that saturated $r$. The $\passrefine$ method either (i) pairs $r'$ to an unsaturated point $r_1 \in R$ (this is done by removing $r'$ from $R$, and adding a new dummy point to $R$ at a very small distance from $r_1$); this makes $r$ unsaturated, or (ii) refines $R$ by removing $|B|$ dummy points (including $r'$) from $R$; again this makes $r$ unsaturated. After this method we have $|B|\neq\emptyset$, $|R|=k'|B|+1$, with $2\leqslant k'\leqslant k$, and $r$ on $\CH{(R-\mathcal{D}_R)\cup B}$; these are necessary conditions for $\procB$.

\begin{algorithm}[H]                 
\caption{\passrefine$(R, B, r)$}          
\label{pass-refine} 
\require{a set $R$ of red points, a non-empty set $B$ of blue points, a saturated point $r\in R$ that is on $\CH{(R-\mathcal{D}_R)\cup B}$, where $|R|=k|B|+1$, with $k\geqslant 2$.}\\
\ensure{$R$, $B$, $r$ where $r$ is unsaturated, and $|R|=k'|B|+1$, with $2\leqslant k'\leqslant k$.}
\begin{algorithmic}[1]
\State $S\gets R-\mathcal{D}_R$ 
\State $r'\gets$ the point of $\mathcal{D}_R$ who saturated $r$
\If {$|S|\neq |\mathcal{D}_R|$} 

\noindent there is an unsaturated red point $r_1$ in $S$. We pair $r'$ with $r_1$; $r$ becomes unsaturated while $r_1$ becomes saturated. Then, we return $R$, $B$, $r$. 
\Else 

\noindent all points of $S$ are saturated. Observe that $|R|=|S|+|\mathcal{D}_R|=2|S|$. Thus, $|R|=2|S|=k|B|+1$, implying that $k$ is an odd number and, thus, $k\geqslant 3$. Then, we remove $|B|$ arbitrary dummy points (including $r'$) from $R$. See Figure~\ref{pass-refine-fig}(b). As a result, we have $|R|=(k-1)|B|+1$ where $k\geqslant 3$; alternatively $|R|=k'|B|+1$ where $2\leqslant k'< k$. Then, we return $R$, $B$, $r$.
\EndIf
\end{algorithmic}
\end{algorithm}

Now we show how to handle each of Cases 1, 2, and 3.

\begin{enumerate}[wide, labelindent=0pt,label={\bf Case \arabic*:}]
\item Assume we called $\procA$ on a blue point $b$, and let $r_1,\dots, r_\alpha$ be the selected consecutive red points where $\alpha\geqslant 2$. We do not care about $r_2,\dots,r_{\alpha-1}$ at this moment as they will be of degree one in the final tree. Before calling $\procB(R_1,B_1, r_1)$ and $\procB(R_2, B_2, r_\alpha)$, we look at the following two cases, depending on whether or not $r_1$ and $r_\alpha$ saturate each other. 
\begin{itemize}
\item {\em $r_1$ is saturated by $r_\alpha$, or $r_\alpha$ is saturated by $r_1$.} In this case $\alpha=2$ and hence $r_\alpha=r_2$. By symmetry, assume $r_1$ is saturated by $r_\alpha$; note that $r_\alpha$ is a dummy point. First we remove the dummy point $r_\alpha$ from $R_2$. Then we add $r_1$ to $R_2$. Note that $|R_2|$ does not change. In addition $r_1$ belongs to both $R_1$ and $R_2$. Then, we proceed as in Case 2 where $\alpha=1$ (the degree of $b$ will be 1 instead of 2).

\item {\em $r_1$ is not saturated by $r_\alpha$, and $r_\alpha$ is not saturated by $r_1$.} We only describe the solution for handling $\procB(R_1, B_1, r_1)$; we handle $\procB(R_2, B_2, r_\alpha)$ in a similar way. We consider three cases, depending on whether $r_1$ is unsaturated, saturated, or dummy. 
\begin{itemize}
\item {\em $r_1$ is unsaturated.} We connect $b$ to $r_1$, and then call $\procB(R_1, B_1, r_1)$.

\item {\em $r_1$ is saturated.} We connect $b$ to $r_1$. Let $r'_1$ be $r_1$'s corresponding dummy point. If $r'_1\notin R_1$, then we call $\procB(R_1, B_1, r_1)$. If $r'_1\in R_1$, then we call $\passrefine(R_1,B_1,r_1)$ first; this makes $r_1$ unsaturated and $|R_1|=k'|B_1|+1$, with $2\leqslant k'\leqslant k$. Then we call $\procB(R_1, B_1, r_1)$.

\item {\em $r_1$ is dummy.} Let $r$ be $r_1$'s corresponding saturated point. We connect $b$ to $r$ (instead of connecting to $r_1$). If $r\notin R_1$ (note that $r=r_2$) then we remove $r_1$ from $R_1$ and add $r$ to $R_1$; this makes $r$ unsaturated in $R_1$ while $|R_1|$ does not change. Then we call $\procB(R_1, B_1, r)$. If $r\in R_1$ then we first call $\passrefine(R_1,B_1,r)$ and then we call $\procB(R_1, B_1, r)$.
\end{itemize}
\end{itemize}
\item $\alpha=1$. Recall that in this case we connected $b$ to a red point $x_s$. Moreover, we looked at two cases depending on whether or not $l_1$ intersects $C_2$ and $l_2$ intersects $C_1$. 

\begin{enumerate}[leftmargin=*,itemindent=30pt, labelindent=0pt,label=Case 2.\arabic*:]
\item {\em $l_1$ does not intersect $C_2$, or $l_2$ does not intersect $C_1$}. By symmetry, we assume $l_1$ does not intersect $C_2$. We consider three cases: 
\begin{itemize}
 \item {\em $x_s$ is unsaturated.} We connect $b$ to $x_s$ and call $\procB(R',B',x_s)$.
 \item {\em $x_s$ is saturated.} We connect $b$ to $x_s$. Then we call $\passrefine(R',B',x_s)$, then call $\procB(R',B',x_s)$.
 \item {\em $x_s$ is dummy.} Let $x$ be the saturated point corresponding to $x_s$. We connect $b$ to $x$. Then, we call $\passrefine(R',B',x)$, then call $\procB(R',B',x)$.
\end{itemize}

\item {\em $l_1$ intersects $C_2$, and $l_2$ intersects $C_1$}.
Recall that in this case $l_1$ intersects either $C'_2$ or $C''_2$. In either case we connected $b$, $x_{s-1}$, $y_i$ to $x_s$, and then called $\procA(R'_1,B'_1,x_{s-1})$ and $\procA(R'_2,B'_2,y_i)$ (assuming $x_s\in R_1$). We have to make sure $b$, $x_{s-1}$, and $y_i$ be connected to an unsaturated red point. We also have to make sure the size conditions for both $\procA(R'_1,B'_1,x_{s-1})$ and $\procA(R'_2,B'_2,y_i)$ are satisfied. In order to do that, we distinguish between two cases: $x_s\neq r_1$ and $x_s=r_1$. In either case, we compute $B'_1$ and $B'_2$ as usual, and we show how to compute $R'_1$ and $R'_2$.
\begin{enumerate}[leftmargin=*,itemindent=40pt, labelindent=0pt,label=Case 2.2.\arabic*:]
\item $x_s\neq r_1$. Before computing $R'_1$ and $R'_2$ we do the following. 
\begin{enumerate}[label=\alph*.]
 \item If $x_s$ is unsaturated, then we connect $b$, $x_{s-1}$,  $y_i$ to $x_s$.
\item If $x_s$ is saturated, then we connect $b$, $x_{s-1}$,  $y_i$ to $x_s$, then call $\passrefine(R_1,B_1,x_s)$. 
\item If $x_s$ is dummy, let $x$ be the saturated point corresponding to $x_s$. We connect $b$, $x_{s-1}$, $y_i$ to $x$, then call $\passrefine(R_1,B_1,x)$. 
\end{enumerate}

At this point we have $|R_1|=k'|B_1|+1$, with $2\leqslant k'\leqslant k$, and $|R_2|=k|B_2|+1$, with $k\geqslant 2$. 
Note that $r_1$ is in both $R_1$ and $R_2$. 
Now we show how to compute $R'_1$ and $R'_2$; if $r_1$ is saturated/dummy, then we have to make sure $r_1$ and its corresponding dummy/saturated point do not lie in different sets. We differentiate between the following cases: 
\begin{itemize}
 \item {\em $r_1$ is unsaturated.} We compute $R'_1$ and $R'_2$ as usual: $R'_1=R_1-\{x_s\}$ and $R'_2=R_2-\{r_1\}$; this makes $|R'_1|=k'|B'_1|$ and $|R'_2|=k|B'_2|$.
\item {\em $r_1$ is saturated $($resp. dummy$)$ and its corresponding dummy $($resp. saturated$)$ point belongs to $R_1$.} We compute $R'_1$ and $R'_2$ as usual: $R'_1=R_1-\{x_s\}$ and $R'_2=R_2-\{r_1\}$.
\item {\em $r_1$ is saturated and its corresponding dummy point, say $r'_1$, belongs to $R_2$.} We compute $R'_1=R_1-\{x_s\}$ and $R'_2=R_2-\{r_1, r'_1\}$. Then, we have $|R'_1|=k'|B'_1|$ and $|R'_2|=k|B'_2|-1=k(|B'_2|-1)+(k-1)$.
\item {\em $r_1$ is dummy and its corresponding saturated point, say $r$, belongs to $R_2$.} We compute $R'_1=(R_1\cup\{r\})-\{x_s,r_1\}$ and $R'_2=R_2-\{r_1, r\}$. Then, we have $|R'_1|=k'|B'_1|$ and $|R'_2|=k(|B'_2|-1)+(k-1)$.
\end{itemize}

In all cases, the size conditions for both $\procA(R'_1,B'_1,x_{s-1})$ and $\procA(R'_2,B'_2,y_i)$ are satisfied. Now we call $\procA(R'_1,B'_1,x_{s-1})$ and $\procA(R'_2,B'_2, y_i)$.

\item $x_s=r_1$.
Observe that in this case $x_s=r_1=y_t$, and $y_i=y_{t-1}$. We distinguish between three cases depending on whether $r_1$ is unsaturated, saturated, or dummy.

\begin{itemize}
 \item {\em $r_1$ is unsaturated.} We connect $b$, $x_{s-1}$, $y_i$ to $r_1$. Then, we compute $R'_1$ and $R'_2$ as usual: $R'_1=R_1-\{r_1\}$ and $R'_2=R_2-\{r_1\}$; this makes $|R'_1|=k|B'_1|$ and $|R'_2|=k|B'_2|$.
\item {\em $r_1$ is saturated}. Without loss of generality assume its corresponding dummy point, $r'_1$, belongs to $R_1$. We connect $b$, $x_{s-1}$, $y_i$ to $r_1$. Then, we compute $R'_1=R_1-\{r_1, r'_1\}$ and $R'_2=R_2-\{r_1\}$; this makes $|R'_1|=k(|B'_1|-1)+(k-1)$ and $|R'_2|=k|B'_2|$.
\item {\em $r_1$ is dummy.} Without loss of generality assume its corresponding saturated point, say $r$, belongs to $R_1$. We connect $b$, $x_{s-1}$, $y_i$ to $r$. Then, we compute $R'_1=R_1-\{r_1, r\}$ and $R_2=R_2-\{r_1\}$; this makes $|R'_1|=k(|B'_1|-1)+(k-1)$ and $|R'_2|=k|B'_2|$. 
\end{itemize}
In all cases, the size conditions for both $\procA(R'_1,B'_1,x_{s-1})$ and $\procA(R'_2,B'_2,y_i)$ are satisfied. Now we call $\procA(R'_1,B'_1,x_{s-1})$ and $\procA(R'_2,B'_2, y_i)$.
\end{enumerate}
\end{enumerate}
\item Recall that $R$ is a set of size $k|B|$ with some dummy points. If there is a blue point, $b$, on $\CH{R\cup B}$ then we simply call $\procA(R,B,b)$. Assume, all points of $\CH{R\cup B}$ are red. Recall that in Case 2 of $\setup$ we choose an arbitrary red point $r$ on $\CH{R\cup B}$. Here we show how to choose $r$ as an unsaturated point; the next steps would be the same as in Case 2 of $\setup$. 
Select a point $r_1$ on $\CH{R\cup B}$. We distinguish between the following three cases:
\begin{itemize}
\item {\em $r_1$ is unsaturated.} We choose $r$ to be $r_1$.
\item {\em $r_1$ is saturated.} First, we pair $r_1$'s corresponding dummy point to an unsaturated point in $R$ that is different from $r_1$ (since we have more red points than blue points at the beginning, such an unsaturated point exists). Then, we choose $r$ to be $r_1$.

\item {\em $r_1$ is dummy.} Let $r_2$ be $r_1$'s corresponding saturated point. First, we pair $r_1$ to an unsaturated point in $R$. Then, we choose $r$ to be $r_2$.
\end{itemize}
Now $r$ is an unsaturated point on $\CH{R\cup B}$; this makes sure that the only dummy point added in this stage will not be in the final tree.
\end{enumerate}
This completes the proof of Conjecture~\ref{conj2}.

\section{Proof of Theorem~\ref{thr0}}
\label{contribution-section}
In this section we prove Theorem~\ref{thr0}: Given two disjoint sets, $R$ and $B$, of points in the plane such that $|B|\leqslant|R|$ and $R\cup B$ is in general position, and let $\delta=\left\lceil\frac{|R|-1}{|B|}\right\rceil$; we prove there exists a plane bichromatic tree on $R\cup B$ whose maximum vertex degree is at most $\max\{3, \delta + 1\}$; this is the best possible upper bound on the maximum degree.

We differentiate between two cases: $|R|<|B|+2$ and $|R|\geqslant |B|+2$.

\begin{enumerate}[wide, labelindent=0pt,label={Case \arabic*:}]
\item $|R|<|B|+2$. In this case $|R|=|B|$ or $|R|=|B|+1$, and hence $\delta = 1$. As shown in Section~\ref{introduction-section}, it may not be possible to find a plane bichromatic tree of degree $2$ (i.e., $\delta+1$), and hence $3$ is the smallest possible degree. On the other hand, when $|R|=|B|$, Kaneko~\cite{Kaneko1998} showed how to compute a $3$-tree on $R\cup B$, and when $|R|=|B|+1$ by Conjecture~\ref{conj2} there exists a $3$-tree on $R\cup B$. This proves Theorem~\ref{thr0} for this case.   

\item $|R|\geqslant |B|+2$. In this case $\delta\geqslant 2$. As shown in Section~\ref{introduction-section}, $\delta+1$ is the smallest possible degree. Since $|B|<|R|$, we have $(k-1)|B|<|R|\leqslant k|B|$, for some $k\geqslant 2$. We distinguish between two cases where $(k-1)|B|+1<|R|\leqslant k|B|$ and $|R|=(k-1)|B|+1$. First assume $(k-1)|B|+1<|R|\leqslant k|B|$, with $k\geqslant 2$. In this range we have $\delta=\left\lceil\frac{|R|-1}{|B|}\right\rceil=\left\lceil\frac{|R|}{|B|}\right\rceil=k$. By Conjecture~\ref{conj2} there exists a $(\delta+1)$-tree on $R\cup B$. This proves Theorem~\ref{thr0} for this case.

Now, assume $|R|=(k-1)|B|+1$, with $k\geqslant 2$. If $k=2$, then $|R|=|B|+1$; we have already proved this case in Case 1. Thus, assume $k\geqslant 3$. Let $k'=k-1$. Then, $|R|=k'|B|+1$, with $k'\geqslant 2$.
In this case $\delta=k-1=k'$, and we have to prove the existence of a $(k'+1)$-tree. If there exists a red vertex $r$ on $\CH{R\cup B}$, then by running $\procB(R,B, r)$, we obtain a $(k'+1)$-tree and we are done. Assume all vertices of $\CH{R\cup B}$ are blue. In order to handle this case, first, we prove the following lemma by a similar idea as in the proof of Lemma~\ref{convex-3blue}.

\begin{lemma}
\label{convex-3blue-2}
Let $R$ and $B$ be two sets of red and blue points in the plane, respectively, such that $|B|\geqslant 1$, $|R|=k'|B|+1$, with $k'\geqslant 2$, and $R\cup B$ is in general position. Let $b_1,b,b_2$ be blue points that are counter clockwise consecutive on $\CH{R\cup B}$. Then, in the radial ordering of $R\cup B-\{b\}$ around $b$, there are $k=k'+1$ consecutive red points, $r_1,\dots,r_k$, such that $|R_1|=k'|B_1|+1$ and $|R_2|=k'|B_2|+1$, where $R_1$ $($resp. $B_1$$)$ is the set of red points $($resp. blue points$)$ of $R\cup B-\{b\}$ lying on or to the left of $\ell(b,r_1)$, and $R_2$ $($resp. $B_2$$)$ is the set of red points $($resp. blue points$)$ of $R\cup B-\{b\}$ lying on or to the right of $\ell(b,r_k)$.
\end{lemma}
\begin{proof}
Assume $b$ is the lowest point of $\CH{R\cup B}$, $b_1$ is to the left and $b_2$ is to the right of $b$. We define the function $f$ as follows: For each point $x$ in $R\cup B-\{b\}$, 
\begin{align*}
f(x)= & k'\cdot (\text{the number of points of $B-\{b\}$ lying on or to the left of $\ell(b,x)$})\\
& -(\text{the number of points of $R$ lying on or to the left of $\ell(b,x)$}).
\end{align*}
Since $b_1$ is the first point and $b_2$ is the last point along the clockwise radial ordering of $R\cup B-\{b\}$ around $b$, we have $f(b_1)= k'$ and $f(b_2)=k'(|B|-1)-|R|=-(k'+1)=-k$. The value of $f$ changes by $+k'$ at every blue point and by $-1$ at every red point. Since $f(b_1)>0>f(b_2)$, there exists a point in the radial ordering for which $f$ equals 0. Let $v$ be the last point in this radial ordering where $f(v)=0$. Let $S=(r_1,\dots, r_k)$ be the sequence of $k$ points strictly after $v$ in the radial ordering. The points of $S$ are red, because if $r_i\in S$ is blue, then $f(r_i)\geqslant 0$ and hence there is a point between $r_i$ and $b_2$ (including $r_i$) in the radial ordering for which $f$ equals 0; this contradicts the fact that $v$ is the last point in the radial ordering with $f(v)=0$. We show that $S$ satisfies the statement of the lemma.
Having $r_1$ and $r_k$, we define $R_1$, $B_1$, $R_2$ and $B_2$ as in the statement of the lemma. We have $f(r_1)=k'|B_1|-|R_1|=-1$, and hence $|R_1|=k'|B_1|+1$. Moreover, 
\begin{align*}
|R_2| &= |R|-|R_1|-k +2\\
&= (k'|B|+1)-(k'|B_1|+1)-(k'+1) +2\\
&= k'(|B|-|B_1|-1)+1\\
 &= k'|B_2|+1.
\end{align*}
\end{proof}

Let $b$ be an arbitrary blue point on $\CH{R\cup B}$. Since all the points on $\CH{R\cup B}$ are blue, by Lemma~\ref{convex-3blue-2} there are $k=k'+1$ consecutive red points $r_1, \dots, r_k$ in the radial ordering of $R\cup B-\{b\}$ around $b$ that divide the point set into two pairs of sets $\{R_1,B_1\}$ and $\{R_2,B_2\}$ with $r_1$ on $\CH{R_1\cup B_1}$ and $r_k$ on $\CH{R_2\cup B_2}$ such that $|R_1|=k'|B_1|+1$, $|R_2|=k'|B_2|+1$. Let $T_1$ and $T_2$ be the plane bichromatic $(k'+1)$-trees obtained by $\procB(R_1, B_1, r_1)$ and $\procB(R_2, B_2,r_k)$, respectively; note that $\dgT{T_1}{r_1}\in\{1,2\}$ and $\dgT{T_2}{r_k}\in\{1,2\}$. Then, we obtain a desired $(k'+1)$-tree $T=T_1\cup T_2\cup \{(b,r_1),\dots,(b,r_k)\}$ with $\dgT{T}{r_1}\in\{2,3\}$, $\dgT{T}{r_k}\in\{2,3\}$, and $\dgT{T}{b}=k=k'+1$. This completes the proof of Theorem~\ref{thr0}.
\end{enumerate}

\section{Conclusion}
\label{conclusion-section}
In this paper, we answered the question posed by Abellanas et al. in 1996~\cite{Abellanas1996}, in the affirmative, by proving the conjectures made by Kaneko in 1998~\cite{Kaneko1998}. In fact we proved a slightly stronger result in Theorem~\ref{thr0}. 

A simple reduction from the convex hull problem shows that the computation of a plane bichromatic spanning tree has an $\Omega(n \log n)$ lower bound. Using a worst-case deletion-only convex hull data structure, we can compute the tree in Theorem~\ref{thr0} in $O(n\cdot \polylog(n))$ time.
\bibliographystyle{abbrv}
\bibliography{Colored-Trees.bib}
\end{document}